\documentclass[10pt,oneside]{amsart}

\usepackage{times}
\usepackage{chicago}

\input dpdefs.sty
\usepackage{graphicx}

\begin{document}

\def\version{22 October 2006}
\def\version{14 April 2008}

\def\co{\text{co}}
\def\rootJn{J_n^{1/2}}
\def\irootJn{J_n^{-1/2}}
\def\siln{\sum_{i\le n}}
\def\that{\hat t}
\def\bp{\mathbf p}
\def\bp{\mathbf p}
\def\bt{\mathbf t}
\def\bA{\mathbf A}
\def\bV{\mathbf V}
\def\bX{\mathbf X}
\def\bY{\mathbf Y}
\def\yb{\mathbf y}
\def\bZ{\mathbf Z}
\def\bz{\mathbf z}
\def\be{\mathbf e_1}
\def\bm{\mathbf m}
\def\bw{\mathbf w}

\def\rootlambda{\sqrt{\mathstrut\lambda}}
\def\rootlambda{\lam^{1/2}}

\def\ZZ{{\mathbb Z}}
\def\bmw{{\mathbf m(\mathbf w)}}
\def\bV{{\mathbf V}}

\def\AND/{\qt{and}\quad}
\def\diam{\text{diam}}

\title[Chromatic numbers]{
Conditioned Poisson distributions\\ and the 
concentration of chromatic numbers\\
}
\date{\today}

\author{John Hartigan, David Pollard and Sekhar Tatikonda\\
Yale University}

\address{
Statistics and Electrical Engineering Departments\\
 Yale University}
 \email{firstname.lastname@yale.edu for each author}
\urladdr{http://www.stat.yale.edu/\~{}ypng/}

\keywords{Random graph, chromatic number, second moment method, 
categorical data, two-way tables, Poisson
counts}

\begin{abstract}\normalsize
The paper provides a simpler method for proving a delicate inequality that was used by Achlioptis and Naor to establish asymptotic concentration for chromatic numbers 
of Erd\"os-R\'enyi random graphs.  The simplifications come from two new ideas.  The first involves a sharpened form of a piece of statistical folklore regarding goodness-of-fit tests for two-way tables of  Poisson counts  under linear conditioning constraints.  The second idea takes the form of a new  inequality that controls the extreme tails of the distribution of a quadratic form in independent Poissons random variables.
\end{abstract}
\maketitle


\section[intro]{Introduction}
Recently, \citeN{AchlioptasNaor05AM} established a most elegant result
concerning colorings of the Erd\"os-R\'enyi random graph, which has vertex set
$V=\{1,2,\dots,n\}$ and has each of the~$\binom n2$  possible edges included independently with probability~$d/n$, for
a fixed parameter~$d$.
  They showed that, as~$n$ tends to infinity, the chromatic number
concentrates (with probability tending to one) on a set of two values, which
they specified as explicit functions of~$d$. The
main part of their argument used the ``second moment
method'' \cite[Chapter~4]{AlonSpencer2000} to establish existence of
desired colorings with probability bounded away from zero.  Most of their
paper was devoted to a delicate calculation bounding the ratio of a second
moment to the square of a first moment.

More precisely,  A\&N considered the quantity
$$
A_n(c) := \frac{n^{k-1}}{k^{2n}}\left(\strut 1-\frac1k\right)^{-2nc}\sum_{\ell\in\hh_k}
\frac{n!}{\prod_{i,j}\ell_{ij}!}\left(\strut 1-\frac2k
+\sum\NL_{i,j}\left(\frac{\ell_{ij}}{n}\right)^2\right)^{nc} ,
$$
where $\hh_k$ denotes the set of all $k\times k$ matrices with nonnegative entries for
which each row and column sum equals $B:=n/k$. (With no loss of generality, A\&N assumed
that~$n$ is an integer multiple of~$k$.) 
 They needed to show, for each fixed $k\ge3$,  that  
\begin{equation}
A_n(c)=O(1)
\qt{when
$c< (k-1)\log(k-1)$.}
\label{main.inequality}
\end{equation}

In this paper we show how the A\&N
calculations can be simplified by
using results about conditioned Poisson distributions. More precisely, we
show that the desired behaviour of~$A_n(c)$ follows from a  sharpening of
a conditional limit theorem due to
\citeN{Haberman74book} together with some elementary facts about
the Poisson distribution. 

In Section~\ref{poisson.facts} we will establish some basic notation and record some elementary facts about the Poisson distribution.
In Section~\ref{heuristics}  we will explain how $A_n(c)$ can be bounded by a conditional expectation of an exponential function of the classical goodness-of-fit statistic for two-way tables. We will outline our proof of~(\ref{main.inequality}), starting from
a $\chi^2$ heuristic that can be sharpened (Section~\ref{condit.poisson}) into a rigorous proof that handles the contributions to~$A_n(c)$ from all except some extreme values of~$\ell$. To control the contributions from the extreme~$\ell$ we will use an inequality (Lemma~\ref{h.ineq}) that captures the large deviation behaviour of conditioned Poissons. The proof of the Lemma (in Section~\ref{delicate}) is actually the most delicate part of our argument.

\section{Facts about the Poisson distribution}  \label{poisson.facts}
Many of the calculations in our paper  involve the  convex function
\begin{equation}
h(t)= (1+t)\log(1+t)-t\qt{for $-1\le t$},
\label{h.def}
\end{equation}
which achieves its minimum value of zero at $t=0$. Near its minimum, $h(t)=t^2/2+O(|t|^3)$.
In fact, $h(t)= \half t^2\psi(t)$ where $\psi$ is a decreasing 
function with $\psi(0)=1$ and $\psi'(0)=-1/3$. See \citeN[page~312]{PollardUGMTP} for a
simple derivation of these facts.  

Define $\NN_0=\{0,1,2,\dots\}$, the set of all nonnegative integers.

\Lemma  \label{Poisson.dist}
Suppose $W$ has a $\Poisson(\lam)$ distribution, with $\lam\ge1$.  
\begin{enumerate}
\ritem(i)
If $\ell=\lam+ \lam u\in\NN_0$ then
\begin{align*}
\sqrt{2\pi\lam}\PP\{W=\ell\} 
&= \exp\left(\strut -\lam h(u) -\tfrac12\log(1+u) +O(1/\ell)\right)\\
&= \exp\left(\strut -\tfrac12\lam{u^2}+O\left(|u|+\lam|u|^3\right) \right)  .
\end{align*}

\ritem(ii)
$\PP\{W=\ell\}\le \exp(-\lam h(u))$ for all
$\ell=\lam(1+u)\in\NN_0$.

\ritem(iii)
For all $w\ge 0$,
\begin{align*}
\PP\{|W-\lam|\ge \lam w\} &\le 2\exp(-\lam h(w))
= 2\exp\left(\strut -\tfrac12\lam{w^2}
+O\left(\strut \lam|w|^3\right) \right)
\end{align*}
\end{enumerate}
\endLemma
\Proof 
By Stirling's formula, 
$$
\log(\ell!/  \sqrt{2\pi}) =(\ell+\tfrac12)\log(\ell)-\ell+r_\ell
\qt{where }\frac1{12\ell+1}\le r_\ell\le \frac1{12\ell}.
$$
Thus
\begin{align*}
\log\bigl(\strut \sqrt{2\pi\lam}\PP\{W=\ell\}\bigr)
&= -\lam +\ell\log(\lam) - \log(\ell!/ \sqrt{2\pi})+\tfrac12\log(\lam)\\ 
&= -\lam h(u)  -\half\log(1+u)+O(\ell^{-1}),
\end{align*}
which gives~(i).

For~(ii), first note that $\PP\{W=0\}=e^{-\lam} =\exp(-\lam h(-1))$.  For $\ell\ge1$ we have
\begin{align*}
\log\left(\strut \sqrt{2\pi}\PP\{W=\ell\}\right) 
&= -\lam +\ell\log(\lam)-(\ell+\tfrac12)\log(\ell)+\ell-r_\ell\\
&\le -\lam +\ell\log(\lam)-\ell\log(\ell)+\ell  =\lam h(u).
\end{align*}

Inequality~(iii) comes from  two appeals to the  usual trick with the moment
generating function
$\PP e^{tW}=
\exp(\lam(e^t-1))$. For $w\ge0$,
$$
\PP\{W\ge\lam+\lam w\} \le \inf_{t\ge0}\PP e^{t(W-\lam-\lam w)}
=\inf_{t\ge0}\exp\left(\strut -t\lam(1+w) +\lam(e^t-1)\right)
$$
The infimum is achieved at  $t=\log(1+w)$, giving the bound~$\exp(-\lam h(w))$. 
Similarly
$$
\PP\{W\le\lam-\lam w\} \le \inf_{t\ge0}\PP e^{t(\lam-\lam w-W)}
=\inf_{t\ge0}\exp\left(\strut t(\lam-\lam w) +\lam(e^{-t}-1)\right)
$$
with the infimum achieved at $t= -\log(1-w)$ if $0\le w<1$ or as $t\to\infty$
if~$w=1$. The inequality is trivial for $w>1$.
\endProof

\section{Heuristics and an outline of the proof of~(\ref{main.inequality})}    \label{heuristics}
We first show that~$A_n(c)$ is almost a conditional
expectation involving a set of
independent random variables,  $Y=\{Y_{ij}:1\le i, j,\le
k\}$, each distributed Poisson$(\lam_{ij})$ with $\lam_{ij}=n/k^2$ for all $i,j$.  For
$\ell\in\hh_k$,
$$
p(\ell) :=\PP\{Y=\ell\} = \frac{e^{-n}(n/k^2)^n}{ \prod_{i,j}\ell_{ij}!}
= \frac{n!}{\prod_{i,j}\ell_{ij}!}\frac{n^{n}e^{-n}}{n!}k^{-2n}
$$
The standardized variables $X_{ij} :=(Y_{ij}-\lam_{ij})/\sqrt{\lam_{ij}}$ are approximately
independent standard normals.

As we  show in Section~\ref{condit.poisson}, the quantity
$$
\b_{n} := n^{(2k-1)/2}\PP\{Y\in \hh_k\} 
$$
converges to a strictly positive constant as $n$ tends to infinity.
Thus
\begin{align*}
p_2(\ell) := \PP\{Y=\ell\mid Y\in\hh_k\} &= p(\ell)/\PP\{Y\in \hh_k\}\\
&=
n^{k-1}k^{-2n}\dfrac{n!}{\prod_{i,j}\ell_{ij}!}
\dfrac{n^{n+1/2}e^{-n}}{ n!\,\b_n}   .
\end{align*}
By Stirling's approximation, the final fraction converges to a nonzero constant.
The quantity $A_n(c)$ is bounded by a constant multiple of 
\begin{equation}
\left(\strut 1-\frac1k\right)^{-2nc}\sum\NL_{\ell\in\hh_k}
p_2(\ell)\left(\strut 1-\frac2k +\sum\NL_{i,j}(\ell_{ij}/n)^2\right)^{nc}    .
\label{bound1}
\end{equation}
That is, for some constant~$C_0$,
$$
A_n(c) \le C_0 \left(\strut 1-\frac1k\right)^{-2nc}\PP_2\left(\strut 1-\frac2k+\sum\NL_{i,j}(Y_{ij}/n)^2\right)^{nc} ,
$$
where $\PP_2(\cdot)$ denotes expectations with respect to  the conditional probability distribution
$\PP(\cdot\mid Y\in\hh_k)$.   

Note the similarily to  the usual chi-squared goodness-of-fit
statistic,
$$
|X|^2 := 
\sumnl_{i,j}X_{ij}^2 = -n+nk^2\sum\NL_{i,j}(Y_{ij}/n)^2  .
$$
The quantity in~\cref{bound1} equals the~$\PP_2$ expectation of
$$
\left(\strut 1-\frac1k\right)^{-2nc} 
\left(\strut \left(\strut 1-\frac1k\right)^2 +\frac{|X|^2 }{nk^2}\right)^{nc}
\le \exp\left(\strut \frac{c|X|^2 }{(k-1)^2}\right).
$$
Our task has become: for a fixed $J_k := c/(k-1)^2 <\rho_k :=\log(k-1)/(k-1)$, show that 
\begin{equation}
\PP_2\exp\left(\strut {J_k|X|^2 }\right)  =O(1)
\qt{as $n\to \infty$. }
\label{main.task}
\end{equation}

Under~$\PP_2$, the random vector~$X$ has a limiting normal distribution~$\nn$ that concentrates on a~$(k-1)^2$-dimensional subspace of~$\RR^{k\times k}$.  The random variable~$|X|^2$ 
has an asymptotic $\chi^2_{R}$ distribution
with $R= (k-1)^2$.
If we could assume that $|X|^2$ were exactly $\chi^2_{R}$-distributed,
we could bound the conditional expectation  in~(\ref{main.task}) by a constant times
$$
\int_0^\infty t^{R/2-1}\exp\left(\strut ct/R -t/2)\right)\,dt  ,
$$
which would be finite  for $c<R/2=(k-1)^2/2$.

To make the argument rigorous we will need to consider the contributions from
the large~$|Y_{ij}-n/k^2|$'s more carefully.  As a special case of Theorem~\ref{normal.bound} in Section~\ref{condit.poisson}, we know that for each fixed~$\theta>1$
there exists a $\del=\del_\th$ for which
\begin{equation}
\PP_2 \exp\left(\strut J_k|X|^2\right)\{|X|\le \del\sqrt n\} \le \th\nn\exp(\th^2
J_k|x|^2).
\label{chi2rigor}
\end{equation}
The expectation with respect to the normal distribution~$\nn$ can be bounded as in the previous paragraph because $|x|^2\sim\chi^2_{(k-1)^2}$ under~$\nn$.

To control the contribution from $\{|X|>\del\sqrt n\}$ it is notationally cleaner to work
with the variables $U_{ij}:= (Y_{ij}-\lam_{ij}) /\lam_{ij}$, that is, $U=kX/\sqrt n$. 
Write $\uu$ for the set of all $u$ in~$\RR^{k\times k}$ for which
$\lam_{ij}(1+u_{ij})\in\NN_0$ for all~$i,j$ and (because~$Y$ is constrained to 
lie in~$\hh_k$),
\begin{equation}
-1\le u_{ij}\le k-1\AND/
\sumnl_i u_{ij} =0 = \sumnl_j u_{ij}   .
\label{u.constraints}
\end{equation}
 We need to bound
\begin{align*}
&\PP \exp\left(\strut nJ_k|U|^2/k^2\right)\{|U|> k\del\}/\PP\{Y\in \hh_k\}  \\
&\qquad= O(n^{(2k-1)/2})\sum\NL_{u\in\uu}\{|u|>k\del\}\PP\{U= u\}\exp(nJ_k|u|^2/k^2)\\
\end{align*}
From Lemma~\ref{Poisson.dist},
$$
\PP\{U=u\} \le \prod\NL_{ij}\exp (-nh(u_{ij})/k^2) ,
$$
which leads us to the task of showing that
\begin{equation}
\sum\NL_{u\in\uu}\{|u|>k\del\}\exp\left(\strut\frac n{k^2}
 \sum\NL_{ij}\left(\strut J_ku_{ij}^2-h(u_{ij})\right)\right) = O(n^{-(2k-1)/2}) .
\label{task}
\end{equation}
Here we can make use of an inequality (proved in Section~\ref{delicate}) that controls the exponent in~(\ref{task}). Recall that $h(t)= (1+t)\log(1+t)-t$ and $\rho_k= \log(k-1)/(k-1)$.

\Lemma  \label{h.ineq}
For each $u=(u_1,\dots,u_k)\in\RR^k$ for which $\sum_j u_j=0$ and $-1\le u_j\le k-1$ for
all~$j$, we have $\sum_j h(u_j)\ge \rho_k \sum_j u_j^2$.
\endLemma

When invoked for the sum over~$j$ for each fixed~$i$, the Lemma bounds~\cref{task} by
$$
\sum\NL_{u\in\uu}\{|u|>k\del\}\exp\left(-n\eps_0|u|^2\right) 
\qt{where $\eps_0 := (\rho_k -J_k)/k^2>0$.}
$$
The set $\{u\in\uu: 2^b k\del< |u| \le 2^{b+1} k\del\}$ has cardinality of order
$O((n2^b)^{k^2})$.  The last sum is less than
$$
O(n^{k^2}) \sum\NL_{b\in\NN_0}\exp\left(\strut k^2 b - n\eps_04^b\right)
$$
which decreases exponentially fast with~$n$.

The  bound asserted in~\cref{main.task} follows.

\section{Limit theory for conditioned Poisson distributions} \label{condit.poisson}
The main result in this Section is Theorem~\ref{normal.bound},  which shows that the contributions to the left-hand side
of~\cref{main.task} from a large range of~$X$ values can actually be bounded using the
$\chi^2$-approximation.

Suppose $\bY=(Y_1,\dots,Y_q)$ is a vector of independent random variables with
$Y_i$ distributed~$\Poisson(\lam_i)$. 
Define
$$
\lam :=(\lam_1,\dots,\lam_q) \AND/ D :=\diag(\rootlambda_1,\dots,\rootlambda_q).
$$
 For the rest of this section assume that~$\nu:=\sum_i\lam_i$ converges to infinity and that 
there exists some fixed constant~$\tau>0$ for which
\begin{equation}
\nu \ge \max\NL \lam_i \ge 
\min\NL_i \lam_i \ge \tau \nu.
\label{min.llam}
\end{equation}
The various constants that appear throughout the section might depend on~$\tau$.

Suppose $\bV_1,\dots,\bV_s$ are fixed vectors in~$\ZZ^q$ that are linearly
independent, spanning a subspace~$\ll$ 
of~$\RR^q$.  The linear independence implies the existence of nonzero constants~$C_1$
and~$C_2$ for which
\begin{equation}
C_1\max\NL_\a|t_\a| \le |\sum\NL_\a t_\a V_\a|\le C_2 \max\NL_\a|t_\a|\qt{for all
$t_\a\in\RR$}.
\label{max.ta}
\end{equation}
We also assume that
\begin{equation}
\ZZ^q\cap(\lambdab\oplus\ll)\ne\emptyset  .
\label{nonempty}
\end{equation}

Under similar assumptions,  
\citeN[Chapter~1]{Haberman74book} proved a central limit theorem for the random vector  
$X:=D^{-1}(Y-\lam)$ conditional on the event
$\{Y\in \lambda\oplus\ll\}$.
The limit distribution~$\nn_\lam$ is that of a~$N(0,I_q)$
conditioned to lie in the $s$-dimensional subspace~$D^{-1}\ll$.  More
precisely,~$\nn_\lam$
has density
$\phi(x)=(2\pi)^{-s/2}\exp(-\norm{x}^2/2)$
with respect to Lebesgue measure $\Leb_\lam$ on the subspace
$D^{-1}\ll$. 

We will write $\QQ(\cdot)$ to denote expectations under~$\PP(\cdot\mid Y\in\lam\oplus
\ll\}$.  That is, for the conditional expectation of a function of~$Y$,
$$
\QQ f(Y) = \frac{\PP f(Y)\{Y\in\lam\oplus\ll\} }{\PP \{Y\in\lam\oplus\ll\} }  .
$$

For the calculations leading to inequality~(\ref{chi2rigor}), the  $q\times1$
vectors  are  more naturally written as $k\times k$ tables.  The vector of means becomes a
table $\lam=\{\lam_{ij}: 1\le i,j\le k\}$ with $\lam_{ij}=n/k^2$ for all $i,j$. The
constraints on row and column sums can be written using the $2k$ tables with ones in a
single row or column, zeros elsewhere, but only $2k-1$ of those tables are linearly
independent.  Thus $q=k^2$ and~$s=k^2-(2k-1)=(k-1)^2$ and $\nu=n$.  The~$\QQ$ in this Section  corresponds to the~$\PP_2$ from Section~\ref{heuristics}.

For each $w\in\ZZ^s$ define $\bz_w :=\sum_{\a\le s}w_\a\bV_\a$, a point
of~$\ZZ^q$. The key idea in Haberman's argument is that the space $\ll$ is
partitioned into disjoint boxes
$$
B_w :=\{\sumnl_{i\le s}t_i\bV_i: \lfloor t_i\rfloor =w_i\}=\bz_w\oplus B_0\qt{for
$w\in\ZZ^s$},
$$
each containing the same number,~$\kappa_V$, of lattice points from~$\ZZ^q$.
Assumption~\cref{nonempty} ensures that
$\kappa_V>0$.

\begin{theorem}  \label{normal.bound}
Suppose $g$ is a uniformly continuous, increasing function. Then for each $\th>1$ there
exists a $\del>0$ and a subset~$\ll_\del$ of~$\ll$ for which
\begin{enumerate}
\ritem(i) $\{x\in D^{-1}\ll: |x|\le \del\sqrt\nu \}\subseteq D^{-1}\ll_\del $
\ritem(ii)
$
\QQ \exp\left(\strut g(|X|^2)\right)\{X\in D^{-1}\ll_\del\} \le \th \nn_\lam
\exp\left(\strut g(\th^2|x|^2)\right)\{x\in D^{-1}\ll_\del/\th\} 
$
\end{enumerate}
\end{theorem}

The proof of the Theorem will be given at the end of this Section, as the culmination of a
sequence of lemmas based on the elementary facts from Section~\ref{poisson.facts}.
We first show that most of  the contributions to the $\PP_2$ and~$\nn_\lam$ probabilities come from a
large, bounded subset of~$\ll$.

\begin{lemma}  \label{tails}
For each $\del>0$ define $\ww_\del :=\{w\in\ZZ^s: \max_\a|w_\a|\le \del \nu\}$ and $\ll_\del
:= \cup_{w\in\ww_\del} B_w$. There exists a constant~$C_\del>0$ for which
$$
\PP\{Y\notin \lam\oplus \ll_\del\} +\nn_\lam(D^{-1}\ll_\del^c)\} = O(e^{-C_\del\nu}) .
$$
\end{lemma}
\begin{proof}
If $y\in\lam\oplus(\ll\bsl\ll_\del)$ then $y-\lam \in \bz_w\oplus B_0$ for some~$w$  with~$\max_\a|w_\a|>\del\nu$, which implies
$$
\sqrt k \max_i|y_i-\lam_i| \ge  |y-\lam| \ge |\bz_w|-\diam(B_0) \ge C_1\del\nu - C_4.
$$
Define $\del_0:= C_1\del/(2\sqrt k)$.
When $\nu$ is  large enough we have $(C_1\del\nu - C_4)/\sqrt k > \del_0\nu \ge \del_0\max_i\lam_i$, so that
$$
\PP\{Y\notin \lam\oplus\ll_\del\} \le \sum\NL_i\PP\{|Y_i-\lam_i|>\del_0\lam_i\} .
$$
Invoke Lemma~\ref{Poisson.dist} to bound
the $i$th summand by $2\exp\left(\strut -\lam_i\del_0^2/2 +O(\del_0^3\lam_i)\right)$.
With a possible decrease in~$\del_0$ we can ensure that the $\lam_i\del_0^2/2$ is at least
twice the other contribution to  the exponent.

Similarly, if $\bx\in D^{-1}(\ll\bsl\ll_\del)$ and $\nu$ is large enough then~$|\bx|
>\del_0\sqrt\nu$ and the  contribution from~$\nn_\lam$ is bounded by a sum of tail
probabilities for the standard normal.
\end{proof}

Next we use Lemma~\ref{Poisson.dist} to get good pointwise approximations for 
$\PP\{Y=\ell\}$ when~$|\ell-\lam|$ is not too large.

\begin{lemma}
\label{pointwise}
For each $\th>1$ there exists a $\del>0$ such that, for all~$\ell=\lam+D x$
in~$\NN_0^q$   for which
$\max_i\lam_i^{-1}|\ell_i-\lam_i|\le \del$,
$$
\th^{-1}\phi\left(\strut
\th \bx\right)\le
\nu^{q/2}\PP\{\bY= \ell\}/\gam(\lam)\\
\le \th \phi\left(\strut
\bx/\th\right)
$$
where  $\gam(\lam) :=
(2\pi)^{s/2}\prod\NL_i(2\pi \lam_i/\nu)^{-1/2}$, a factor that stays bounded away from zero
and infinity as~$\nu\to\infty$.
\end{lemma}
\begin{proof}
From Lemma~\ref{Poisson.dist},
$$
\left(\strut \prod\NL_i\sqrt{2\pi\lam_i}\right)\PP\{\bY= \ell\}  
= \prod\NL_i \exp\left(\strut -\tfrac12
x_i^2+r_i\right)
$$
where, for some constant~$C_3$,
$$
|r_i|\le C_3(|x_i|+|x_i|^3)/\sqrt\nu \le C_3\del(1+x_i^2)  .
$$
The asserted inequalities follow if $\del$ is small enough.
\end{proof}

Next we sum over the pointwise approximations to get bounds for the probability that~$Y$
lies in one of the boxes that partition~$\lam\oplus\ll$.  The sum for the box~$\lam\oplus
B_w$ will run over the lattice points of the form $\lam+D x$ with $x$ in the set
$$
\xx_w =\{x\in D^{-1} B_w: \lam+D x\in\NN_0^q\}  .
$$


\begin{lemma}  \label{box}
For each $\th>1$ there exists a $\del>0$ such that, for all $w$ in~$\ww_\del$ and $\nu$
large enough,
\begin{align*}
\th^{-1} \nn_\lam\left(\strut D^{-1}B_w\th\right)
&\le \nu^{(q-s)/2}\PP\{Y\in \lam\oplus B_w\}/\b(\lam)\\
& \le \th \nn_\lam\left(\strut D^{-1}B_w/\th\right)
\end{align*}
where $\b(\lam)$ is a factor that stays bounded away from zero
and infinity as~$\nu\to\infty$.
\end{lemma}

\begin{proof}
As the proofs for the two inequalities are  similar, we consider only the upper bound.

Define $\bx_w :=D^{-1}\bz_w$.  By inequality~\cref{max.ta} we have $|\bz_w|\le
C_2\del\nu$ and hence $|\bx_w|\le C_4\del\sqrt\nu$ for some constant~$C_4$.  Similarly, for
each
$\yb = \lam +D\bx$ in~$\lam\oplus B_w$ we have $|y-\lam-\bz_w|$ bounded by a constant,
which implies  $|\bx-\bx_w|\le C_5/\sqrt\nu$ and hence
$$
|\,|\bx|^2-|\bx_w|^2\,| \le \del_0:=C_5^2/\nu + 2(C_5/\sqrt\nu)C_4\del\sqrt\nu  .
$$
It follows that for each $\eps>0$ and $\sig$ close enough to~$1$,
$$
\sup\{|\phi(\bx/\sig)/\phi(\bx_w/\sig)-1|: \bx\in D^{-1}B_w\} < \eps
$$
if $\nu$ is large enough and $\del$  is small enough.

Taking~$\sig$ equal to the $\th$ from Lemma~\ref{pointwise} we then have
\begin{align*}
\PP\{Y\in \lam\oplus B_w\}
&=  \sum\NL_{x\in\xx_w}\PP\{Y=\lam+Dx\}\\
&\le \th \gam(\lam)\nu^{-q/2}\sum\NL_{x\in\xx_w}\phi(x/\th)\\
&\le \th \gam(\lam)\nu^{-q/2}\kappa_V (1+\eps)\phi(\bx_w/\th).
\end{align*}
Similarly,
\begin{align*}
\nn_\lam(D^{-1}B_w/\th) &=
\int \{\th \bt\in D^{-1}B_w\}\phi(\bt)\Leb_\lam(d\bt)\\
&= \th^{-s}\int \{\bx\in D^{-1}B_w\}\phi(\bx/\th)\Leb_\lam(d\bx)\\
&\ge \th^{-s}\phi(\bx_w/\th)(1-\eps)\Leb_\lam (D^{-1}B_0).\\
\end{align*}
The invariance properties of Lebesgue measure imply  existence of some
function~$\mu(\lam)$ that stays bounded away from zero and infinity as~$\nu$ tends to
infinity, for which
 $\Leb_\lam (D^{-1}B_0)=\nu^{-s/2}\mu(\lam)$.
Thus
$$
\PP\{Y\in \lam\oplus B_w\} \le \th^{s+1}\frac{1+\eps}{1-\eps}\nu^{-(q-s)/2}
\frac{\kappa_V\gam(\lam)}{\mu(\lam)}
\nn_\lam(D^{-1}B_w/\th) .
$$
Choose $\eps$ small enough and
replace~$\th$ by a value closer to~$1$ to get the  upper half of the asserted
inequality, with $\b(\lam)=\kappa_V\gam(\lam)/\mu(\lam)$.
\end{proof}

\begin{corollary}  \label{hyperplane}
$\PP\{Y\in\lam\oplus\ll\}= \nu^{-(q-s)/2}\left(\strut \b(\lam)+o(1)\right)$
\end{corollary}

\begin{proof}
From Lemmas~\ref{tails} and~\ref{box}, for each~$\th>1$,
\begin{align*}
\PP\{Y\in \lam\oplus\ll\}
&= \PP\{Y\notin\lam\oplus\ll_\del\} +\sum\NL_w\{w\in\ww_\del\}\PP\{Y\in\lam\oplus B_w\} \\
&\le O(e^{-C_\del\nu})
+\th\b(\lam)\nu^{-(q-s)/2}\sum\NL_w\{w\in\ww_\del\}\nn_\lam\left(\strut
D^{-1}B_w/\th\right)\\ 
&\le \th\nu^{-(q-s)/2}\left(\strut \b(\lam)+o(1)\right)
\end{align*}
The argument for the lower bound is similar.
\end{proof}

\begin{corollary}  \label{condit.normal.approx}
For all $\nu$ large enough,
$$
\th^{-1}\nn_\lam(D^{-1}B_w\th) \le \QQ\{Y\in \lam\oplus B_w\}\le 
\th\nn_\lam(D^{-1}B_w/\th) 
$$
 for all~$w\in\ww_\del$.
\end{corollary}

We now have all  the facts needed to prove Theorem~\ref{normal.bound}.
The argument is a slight modification of the method used to prove Lemma~\ref{box}. Start
with the~$\del$ and~$\ll_\del$ from that Lemma. Assertion~(i), modulo an unimportant
constant, was established at the start of the proof of the Lemma.

Define $f(x):=\exp(g(|x|^2))$.  From the proof of the Lemma we know that
$
|\,|\bx|^2-|\bx_w|^2\,| \le \del_0
$.
By uniform continuity of~$g$,
if $\del$ is small enough we then have
$$
|g(|x|^2/\sig^2) -g(|x_w|^2/\sig^2)| <\eps\qt{all $x\in D^{-1}B_w$, all $\sig\approx1$}
$$
and hence
$$
e^{-\eps}f(x_w/\sig) \le f(x/\sig) \le e^\eps f(x_w/\sig)
\qt{all $x\in D^{-1}B_w$, all $\sig\approx1$.}
$$
Use the bounds on~$f$ on~$D^{-1}B_w$ to deduce that
\begin{align*}
\QQ f(X)\{Y\in \lam\oplus  B_w\}
&\le e^\eps f(x_w)\QQ \{Y\in \lam\oplus  B_w\}\\
&\le e^\eps f(\th x_w)\th\nn_\lam(D^{-1}B_w/\th) \qt{as $g$ is increasing}\\
&\le  
e^{2\eps}\th \nn_\lam f(\th x)\{x\in D^{-1}B_w/\th\}
\end{align*}
Sum over $w$ in~$\ww_\del$. to complete the argument.

\section{Proof of Lemma~\ref{h.ineq}}  \label{delicate}
At a key step in the argument we will need the inequality
\begin{equation}
\psi(t)\ge 2\log(1+2t)/(1+2t)\qt{for all $t\ge 0$},
\label{psi.lower}
\end{equation}
for which, unfortunately, we have no direct analytic proof. 
However, the assertion is trivially true  near the origin because the lower bound tends to
zero as~$t$ tends to zero. For large~$t$ the ratio of $\psi(t)$ to the lower bound tends
to~$2$. For intermediate values we have only a  proof based on an analytic bound on
derivatives together with  numerical calculation on a suitably fine grid.  It would be
satisfying to have a completely analytic proof for~\cref{psi.lower}.

Define $g_k(s) :=h(s)-\rho_k s^2$.  We need to show that the function $G_k(u) :=
\sumnl_{j\le k} g_k(u_j)$ is nonnegative on the constraint set.  Suppose the
minimum is achieved at $t=(t_1,\dots,t_k)$. Without loss of generality, we may suppose
$-1\le t_1\le t_2\le \dots\le t_k\le k-1$. We cannot
have~$t_1 =-1$ because $h'(-1)=\infty$. Indeed,
$$
g_k(t_1+\eps)+g_k(t_k-\eps) -g_k(t_1)-g_k(t_k)=\eps\log\eps +O(\eps),
$$
which would be negative for small $\eps>0$.   It then
follows that  $t_k<k-1$ for otherwise the constraint $\sum_j t_j=0$ would force $t_j=-1$ for
$j<k$.

Use Lagrange multipliers (or argue directly regarding the first order effects of 
perturbations $\eps$ with $\sumnl_j\eps_j=0$) to deduce existence of some constant~$\th$
for which 
$g'_k(t_j) =\th$ for all~$j$.

Note that $g'_k(s)=\log(1+s)-2\rho_k s$ is concave (because $g''(s)$ is
decreasing) with $g'_k(0)=0$ and $g''_k(0)=1-2\rho_k>0$.  It follows that
$\th\le 0$ and that there are numbers $-1<a_\th\le 0\le b_\th<k-1$ with
$g'_k(a_\th)=\th=g'_k(b_\th)$ such that $t_j$ equals~$j=a_\th$ for $j\le k-r$ and~$b_\th$
otherwise.   That is $(k-r)a_\th+r b_\th=0$ and
$
G_k(t) = rg_k(b_\th)+(k-r)g_k(a_\th)
$.
 Thus it
suffices for us to show that the functions
$$
M_{r,k}(b) := r g_k(b) + (k-r)g_k(-rb/(k-r))
\qt{for $0\le b < (k-r)/r$}
$$
are nonnegative  for $r=1,2,\dots,k-2$.

For $r\ge2$ and $0\le b\le (k-2)/2$, inequality~\cref{psi.lower} shows that~$g_k(b)$ is
nonnegative:
\begin{align*}
g_k(b) &=  b^2 \left(\strut \frac12\psi(b)-2\rho_k\right)
\ge b^2 \left(\strut \dfrac{\log(1+2b)}{1+2b}-\dfrac{\log(k-1)}{k-1}\right)\ge0
\end{align*}
because $b\mapsto \log(1+2b)/(1+2b)$ is a decreasing function. 
The function~$M_{r,k}(b)$ is then a sum of nonnegative functions on~$[0,(k-r)/r]$.

It remains only to consider the case where~$r$ equals~$1$.
To simplify notation, write $k_1$ for~$k-1$ and abbreviate~$M_{1,k}$ to~$M_k$.
That is,
\begin{align*}
M_k(b) &= h(b) +k_1h(-b/k_1)-\rho_k\left(\strut b^2+k_1(b/k_1)^2\right) \\
&= (1+b)\log(1+b) +(k_1-b)\log(1-b/k_1) - k\rho_kb^2/k_1  ,
\end{align*}
whence
\begin{align*}
M_k'(b) &= \log\left(\strut \frac{1+b}{1-b/k_1}\right)-\frac{2k\rho_kb}{k_1}, \qquad
M_k''(b) = \frac{k}{(1+b)(k_1-b)}-\frac{2k\rho_k}{k_1}.
\end{align*}
Notice that $M_k''(b)\ge0$ except on an interval ~$I_k :=(b_k,b_k')$ in which the
inequality 
$2(1+b)(k_1-b) > k_1/\rho_k$ holds.

For $k=3$ or~$4$ the interval~$I_k$ is empty. The functions~$M_3$ and $M_4$
are convex. They achieve their minima of zero at~$b=0$ because $M_k'(0)=0$.

For~$k\ge 5$, the interval~$I_k$ is nonempty. The
derivative~$M_k'(b)$  achieves its maximum value at $b=b_k$ and its the
minimum value at~$b_k'$.  For $k=5$ we have $b_5'\approx 2.19$ and
$M_5'(b_5')\approx 0.055$.  Thus $M_5$ is an increasing function on~$[0,4]$,
achieving its minimum value of zero at~$b=0$.  

\begin{figure}[htb]
\centerline{
\includegraphics[width=5.5in]{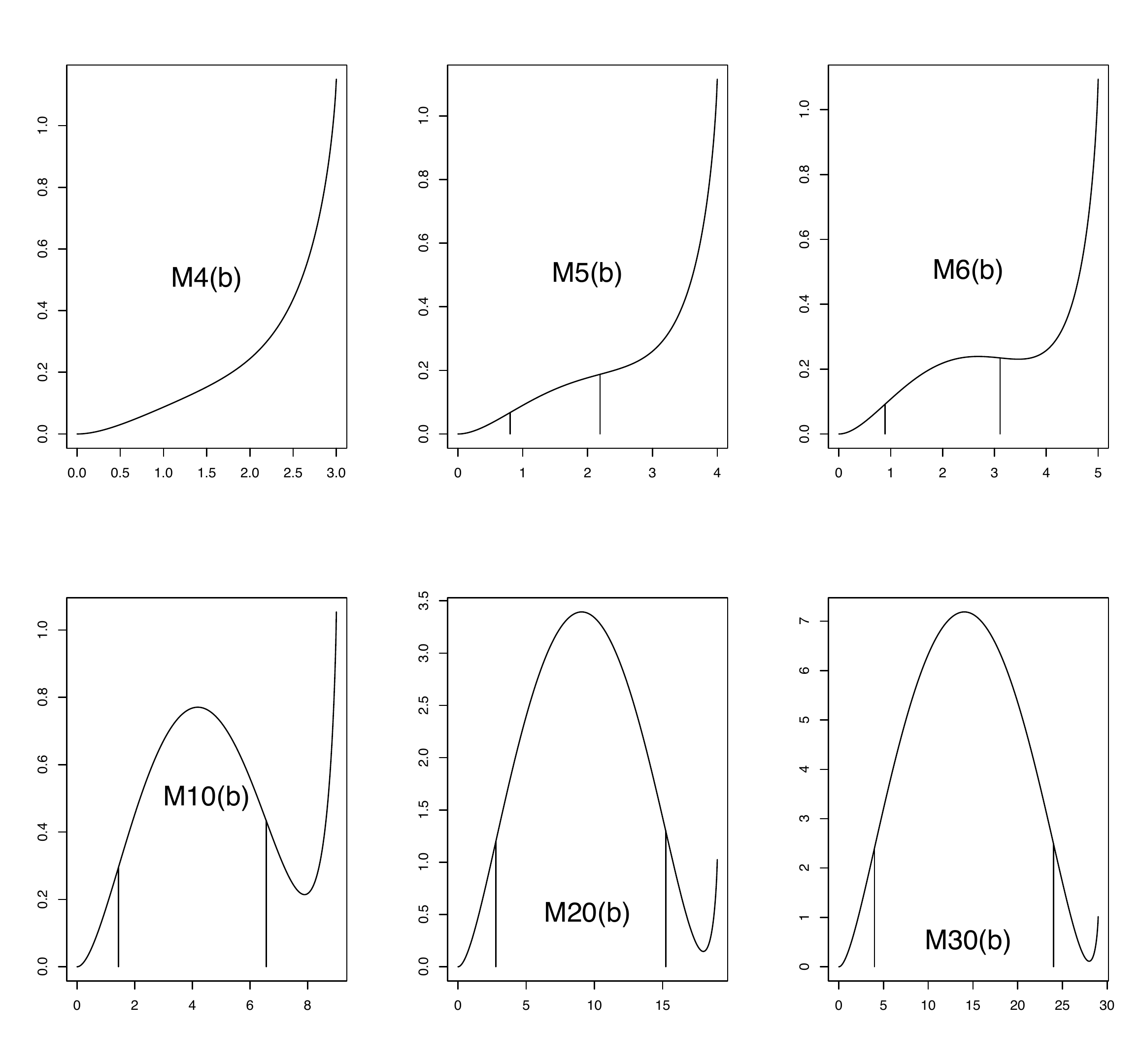}}
\caption{Plots of $M_k(b)$ for various values of~$k$.  The vertical lines mark off the
intervals~$I_k$  where the functions are concave.}
\end{figure}

For $k\ge6$ a more delicate
analysis is required.
The function~$M_k$ is concave on the segment~$I_k$ and
convex on each of the segments $[0,b_k]$ and~$[b_k',k_1]$. The global
minimum is achieved either at~$b=0$, with $M_k(0)=0$, or at the local
minimum~$b^*\in(b_k',k_1)$ where~$M_k'(b^*)=0$ and~$M_k''(b^*)>0$. From the change in sign
of the derivative,
\begin{align*}
M_k'(k_1-1) &= 2\log(k_1)/k_1^2 >0\qt{for all $k$}\\
M_k'(k_1-2) &= \frac{2(k_1+2)\log(k_1)}{k_1^2} +\log\left(\strut \frac{k_1-1}{2k_1}\right)<0
\qt{for
$k\ge6$},
\end{align*}
we deduce that $k_1-2<b^*<k_1-1$. The convexity of~$M_k$ on~$[b_k',k_1]$ then gives a
linear lower bound,
\begin{align*}
M_k(b^*) &\ge M_k(k_1-1) +(b^*-k_1+1)M_k'(k_1-1)\\
&\ge \dfrac{k_1-1}{k_1^2}\log(k_1) - \dfrac{2\log(k_1)}{k_1^2}   \\
&\ge 0 \qt{for $k\ge4$}.
\end{align*}
It follows that $M_k$ is nonnegative also for $k\ge6$.

\bibliographystyle{chicago}
\bibliography{DBP}

\end{document}